\documentclass[showpacs,nofootinbib,showkeys,eqsecnum,prd,aps,preprint]{revtex4-1}

\usepackage[english]{babel}
\usepackage[cp1251]{inputenc}
\usepackage[T1]{fontenc}

\usepackage{amssymb}
\usepackage{amsmath}
\usepackage{amsfonts}
\usepackage{latexsym}
\usepackage{bm}
\usepackage{tensor}
\usepackage{hyperref}

\newtheorem{lemma}{Lemma}

\newtheorem{thr}{Theorem}

\newenvironment{proof}[1][Proof]
{\textbf{#1} }{\ \rule{0.5em}{0.5em}}

\begin{document}

\title{Vacuum quantum effects on  Lie groups with bi-invariant metrics}

\author{A. I. Breev }
\email{breev@mail.tsu.ru}
\affiliation{Department of Theoretical Physics, Tomsk State University Novosobornaya
Sq. 1, Tomsk, Russia, 634050}

\author{A. V. Shapovalov }
\email{shpv@phys.tsu.ru}
\affiliation{Department of Theoretical Physics, Tomsk State University Novosobornaya
Sq. 1, Tomsk, Russia, 634050}
\affiliation{Tomsk Polytechnic University,  Lenin ave., 30, Tomsk, Russia, 634034}

\begin{abstract}
We consider the effects of vacuum polarization and particle creation of a scalar field on Lie groups with a non-stationary bi-invariant metric of the Robertson--Walker type.
The vacuum expectation values of the energy momentum tensor for a scalar field determined by the group representation are found using the noncommutative integration method for the field equations instead of separation of variables. The results obtained are illustrated by the example of the three-dimensional rotation group.
\end{abstract}



\maketitle

\section{Introduction}

Quantum effects of vacuum polarization and particle creation are considered in cosmological models of general relativity, in the theory of super-strong gravitational fields arising in the vicinity of black holes and neutron stars \cite{Grib,Devis}.

For calculations of quantum effects and for clarifying some substantial issues in gravity and in cosmological models various geometric and symmetry ideas, methods and approaches are fruitful. Among them, we note significant progress in the analysis of extended cosmological models achieved using the Noether symmetries (see the review papers \cite{nojiri-odints,capoz-laur-odints2012} and references therein, as well as \cite{capoz-laur-odints2014,capoz-dialekt2018}). Conformal symmetry was used in study of teleparallel gravity models \cite{bamba-odints}.

Cosmological solutions of the Einstein field equations are often found in the class of metrics that admit a certain symmetry group \cite{Ryan,Stephani}, amounting to an assumption of spatial homogeneity. The presence of symmetry greatly simplifies the description of cosmological models and in most cases allows one to advance in the study of quantum effects \cite{Grib,Devis}.

Following these ideas, we consider a group manifold $M=\mathbb{R}^{1}\times G$ where $G$ is a Lie group with a bi-invariant metric.

A Robertson--Walker type metric defined on the manifold $M$ naturally introduces time in $M$, and the group structure provides high symmetry and nontrivial topology of the space-time. This space-time is homogeneous and isotropic, and it differs from the well-studied Friedmann spaces (Friedmann--Robertson--Walker (FRW) space-time). Moreover, this is an appropriate model where the relationship between algebraic and topological characteristics of the space-time with quantum-field effects could manifest themselves more clearly.

Note that the Lie algebra $\mathfrak{g}$ of the Lie group $G$ defines coajoint orbits but it does not completely define the topology of $G$. For a given Lie algebra, the several models of Robertson--Walker type defined on the group manifold $M$ can be constructed that differ from each other in their spatial topology. The Lie group approach allows us to expand the class of cosmological models under consideration which are characterized by different spatial topologies, whereas for the flat Robertson--Walker model spatial topologies are strictly limited.

The aim of this paper is consideration of vacuum polarization and particle creation of the scalar field in the FRW-type space-time $M$ from the group-theoretical point of view. We explore quantum effects in the framework of the one-loop approximation when the quantum scalar field is considered on the background of the classical gravitation field.

Computations of the quantum effects require complete basis of solutions to quantum wave equations that usually implies the method of separation of variables. Unlike this approach, we use the method of noncommutative integration (MNI) of linear partial differential equations proposed in \cite{spK,nonc95}. The MNI allows us to construct a basis of solutions to a wave equation (in particular, the Klein--Gordon equation) avoiding the separation of variables procedure, and using the symmetry algebra $\mathfrak{g}$ of the equation when $\mathfrak{g}$ is the Lie algebra of the Lie group $G$ of $M$.

To study the phenomenon of particle creation in gravitational fields, there exist several approaches, e.g., the Feynman path integral technique \cite{Duru}, an approach that uses the Green functions \cite{Gav}, the semiclassical WKB approximation \cite{Biswas}, and the Hamiltonian diagonalization method \cite{pavlov2016,pavlov2002,pavlov2001} that we use in this paper.

Application of the group-theoretic approach allows one to investigate the effect of symmetry and nontrivial space topology on the particle creation and vacuum polarization. In this regard, we note that influence of space topology on the Casimir effect in the geometry of two parallel plates and in cylindrical carbon nanotubes was studied in \cite{top2011}. The vacuum polarization effect on a Lie group with a stationary right-invariant metric was considered in detail in \cite{key-1,brunim,br2011}.

Here, we study the additional properties of vacuum expectation values, which are characteristic for the case of a bi-invariant metric.

The paper is organized as follows. In Section \ref{sec:Lie-groups} we introduce necessary notations and define geometric characteristics of a space-time conformally equivalent to $M$. In Section \ref{sec:-representation} we describe a special irreducible $\lambda$-representation of a Lie algebra (see Refs. \cite{nonc95,spK}). In terms of the representation introduced, we describe in Section \ref{sec:Noncommutative-integration} the procedure for constructing the basis of solutions to the Klein-Gordon equation in the framework of the MNI. The basis of solutions obtained is used to study the vacuum expectation values of the energy-momentum tensor (EMT) in Section \ref{sec:Vacuum-average}. Section \ref{sec:Particle-creation} presents the study of particle creation in an external gravitational field in the framework of the Hamiltonian diagonalization method in spirit of \cite{Grib,pavlov2001}. In Section \ref{sec:SO3} the general expressions obtained in previous sections are illustrated by an example of $G = SO(3)$. This space-time corresponds to the cosmological model of Bianchi IX in which the diagonal components of the metric are equal each other \cite{ellis,pritomanov}. In Section \ref{sec:Conclusion} we discuss the results obtained.

\section{Lie groups with a bi-invariant metrics of Robertson--Walker type \label{sec:Lie-groups}}

Let $G$ be a real compact semisimple $(n-1)$ - dimensional
Lie group, and $\mathfrak{g}$ is its Lie algebra. Given an element
$X\in\mathfrak{g}$, the adjoint action of $X$ on $\mathfrak{g}$
is the map $\mathfrak{\mathrm{ad_{X}:}\, g\rightarrow g}$ with $\mathfrak{\mathrm{ad_{X}Y=[X,Y]}}$
for all $Y\in\mathfrak{g}.$ The Killing form $\mathbf{\gamma}(X,Y)=k\cdot\mathrm{Tr}(\mathrm{ad}_{X}\cdot\mathrm{ad}_{Y})$,
where $k$ is real parameter, and $X,Y\in\mathfrak{g}$, defines a
bi-invariant Riemannian metric on $G$,
\begin{gather*}
\langle u,w\rangle_{g}=\mathrm{\gamma}\left((L_{g^{-1}})_{*}u,(L_{g^{-1}})_{*}w\right)=\gamma\left((R_{g^{-1}})_{*}u,(R_{g^{-1}})_{*}w\right),
\end{gather*}
 where $(L_{g})_{*} $ and $(R_{g})_{*}$ are differentials of the left- and
the right-  shifts on the Lie group $G$, respectively; $u,w\in T_{g}G,\, g\in G$.

Consider the Robertson--Walker conformal space $\tilde{M}$ with the metrics
\[
d\tilde{s}^{2}=a^{2}(\tau)ds^{2},\quad ds^{2}=d\tau^{2}-dl_{G}^{2},
\]
where $ds^{2}$ is a metric on the group manifold $M=\mathbb{R}^1\times G$
given by the 2-form $\mathbf{G}=1\oplus(-\mathbf{\gamma})$, and
$\tau\in[0;\infty)$ is the conformal time of a comoving observer,
the scale factor $a(\tau)$ is a smooth real function.

The Ricci tensor $\tilde{\mathbf{R}}(X,Y)$ and the scalar curvature
$\tilde{R}$ of the space-time $\tilde{M}$ are related to the
Ricci tensor $\mathbf{R}(X,Y)$ and the scalar curvature $R$ on $M$
as
\cite{Devis}:
\begin{gather}\nonumber
\tilde{\mathbf{R}}(X,Y)=\mathbf{R}(X,Y)-(n-2)(c(\tau)^{2}-\dot{c}(\tau))e_{0}(X)e_{0}(Y)+\\
+\mathbf{G}(X,Y)((n-2)c^{2}(\tau)+\dot{c}(\tau)),\quad X,Y\in\mathbb{R}^{1}\times\mathfrak{g},\nonumber \\ \nonumber
\tilde{R}=a^{-2}(\tau)\left(R+(n-1)(n-2)c^{2}(\tau)+2(n-1)\dot{c}(\tau)\right),\quad c(\tau)\equiv\frac{\dot{a}(\tau)}{a(\tau)},
\end{gather}
where $\dot{a}(\tau)=d a(\tau)/d\tau$, $e_{0}(X)=X^{0}=\langle X,\partial_{\tau}\rangle$
is a projection from $\mathbb{R}^{1}\times\mathfrak{g}$ to $\mathbb{R}^{1}$.
The value of the Ricci tensor for vector fields $u,v\in T_{(\tau,g)}M$
is found using right shifts on the group manifold $M$:
\[
R(u,w)=\mathbf{R}\left((R_{g^{-1}})_{*}u,(R_{g^{-1}})_{*}w\right).
\]
 As the Lie group $G$ with a bi-invariant metric is a space of
constant curvature, then the following relation holds \cite{barut}:
\begin{equation}
\mathbf{R}(X,Y)=-\frac{R}{n-1}\gamma(X,Y),\quad R=-\frac{n-1}{4k}\quad X,Y\in\mathfrak{g}.\label{RabeqR}
\end{equation}
Thus, the structure constants of the Lie algebra $\mathfrak{g}$,
the scale factor $a(\tau)$ and the parameter $k$ define all
geometric characteristics of the space-time $\tilde{M}.$

\section{$\lambda$-representation of a Lie groups\label{sec:-representation}}

A Lie group $G$ acts on a dual space $\mathfrak{g}^{*}$ by a coadjoint
representation $Ad^{*}:G\times\mathfrak{g}^{*}\rightarrow\mathfrak{g}^{*}$
that stratifies $\mathfrak{g}^{*}$ into coajoint orbits \cite{kirr}. We call the coajoint orbits
of maximal dimension, equal to $\mathrm{dim}\mathfrak{g}-\mathrm{ind}\mathfrak{g}$,
non-degenerate. The algebra index, $\mathrm{ind}\mathfrak{g}$, is
defined as the number of independent Casimir functions $K_{\mu}(f)\text{, \ensuremath{f}\ensuremath{\in\mathfrak{g}^{*}}}$,
$\mu=1,\dots,\mathrm{ind}\mathfrak{g}$, on the dual space $\mathfrak{g}^{*}$
with respect to the Poisson -- Lie bracket
\[
\{\psi_{1},\psi_{2}\}(f)=\langle f,\left[\nabla\psi_{1}(f),\nabla\psi_{2}(f)\right]\rangle,\quad\psi_{1},\psi_{2}\in C^{\infty}(\mathfrak{g}^{*}),\quad f\in\mathfrak{g}^{*},
\]
where $[\cdot,\cdot]$ is a commutator in the Lie algebra $\mathfrak{g}$
and $\langle\cdot,\cdot\rangle$ denotes the natural pairing between
the spaces $\mathfrak{g}^{*}$ and $\mathfrak{g}$.

Let $\mathcal{O}_{\lambda}$ be a non-degenerate coajoint orbit passing through
a general covector $\lambda\in\mathfrak{g}^{*}$. Locally, one can
always introduce the Darboux coordinates $(p,q)\in P\times Q$ on the
orbit $\mathcal{O}_{\lambda}$ in which the Kirillov form $\omega_{\lambda}$
defining a symplectic structure on the coajoint orbits has the canonical form
$\omega_{\lambda}=dp^{a}\wedge dq_{a}$, $a=1,\dots,\mathrm{dim}\,\mathcal{O}_{\lambda}/2$.

Denote by $\mathfrak{g}_{\mathbb{C}}$ a complex extension of the
Lie algebra $\mathfrak{g}$. The canonical embedding $f$: $\mathcal{O}_{\lambda}\rightarrow\mathfrak{g}_{\mathbb{C}}^{*}$
is uniquely determined by the functions $f_{X}(p,q,\lambda),\ X\in\mathfrak{g}$,
satisfying the system of equations
\[
\{f_{X},f_{Y}\}=f_{[X,Y]},\quad f_{X}(0,0,\lambda)=\lambda(X),\quad K_{\mu}(f(p,q,\lambda))=K_{\mu}(\lambda),\quad X,Y\in\mathfrak{g}.
\]
We consider the functions $f_{X}(p,q,\lambda)$ to be linear
in the variables $p_{a}$:
\begin{gather}
f_{X}(p,q,\lambda)=\alpha_{X}^{a}(q)p_{a}+\chi_{X}(q,\lambda),\quad q\in Q,\quad p\in P.\label{canf}
\end{gather}
The vector fields $\alpha_{X}(q)=\alpha_{X}^{a}(q)\partial_{q^a}$ are
generators of a transformation group $G_{\mathbb{C}}=\exp(\mathfrak{g}_{\mathbb{C}})$
of a homogeneous space $Q$, and the functions $\chi_{X}(q,\lambda)$
realize \emph{a non-trivial extension} of the vector fields $\alpha_{X}(q)$
\cite{prolong,bardef}. Note that for non-degenerate coajoint orbits there
 always exist the canonical embedding functions having the form (\ref{canf}).

We introduce a measure $d\mu(q)$ and a scalar product
\begin{equation}
(\psi_{1},\psi_{2})=\int_{Q}\overline{\psi_{1}(q)}\psi_{2}(q)d\mu(q),\quad d\mu(q)=\rho(q)dq\label{sp_Q},
\end{equation}
in the space of functions $L_{2}(Q,d\mu(q))$ on a partially holomorphic
manifold $Q$. Here $\overline{\psi_{1}(q)}$ denotes the complex conjugate to $\psi_{1}(q)$.
The first-order operators
\begin{gather}
\ell_{X}(q,\lambda)=if_{X}(-i\partial_{q},q,\lambda)=\alpha_{X}^{a}(q)\partial_{q^{a}}+i\chi_{X}(q,\lambda),\quad K_{\mu}(-i\ell(q,\lambda))=K_{\mu}(\lambda)\label{lrpr}
\end{gather}
realize, by definition, an irreducible $\lambda$-representation of a Lie algebra
$\mathfrak{g}$ in $L_{2}(Q,d\mu(q))$ and are the result of $qp$-quantization
on the coajoint orbit $\mathcal{O}_{\lambda}$ \cite{spK,nonc95,kirr}.

Without loss of generality, we assume that the operators $-i\ell_{X}(q,\lambda)$
are Hermitian with respect to the scalar product (\ref{sp_Q}).

We define the generalized functions $D_{q\overline{q'}}^{\lambda}(g)$
as a solution to the system of equations
\begin{gather}
\left(\eta_{X}(g)+\ell_{X}(q,\lambda)\right)D_{q\overline{q'}}^{\lambda}(g)=0,\quad\left(\xi_{X}(g)+\overline{\ell_{X}(q',\lambda)}\right)D_{q\overline{q'}}^{\lambda}(g)=0,\label{tD}
\end{gather}
 where $\xi_{X}(g)=(L_{g})_{*}X$ and $\eta_{X}(g)=-(R_{g})_{*}X$
are left- and right- invariant vector fields on a Lie group $G$,
respectively.

The functions $D_{q\overline{q'}}^{\lambda}(g)$ provide
the lift of the $\lambda$-representation of the Lie algebra $\mathfrak{g}$
to the local unitary representation $T^{\lambda}$ of its Lie group
$G$,
\begin{equation}\nonumber
(T_{g}^{\lambda}\varphi)(q)=\int_{Q}D_{q\overline{q'}}^{\lambda}(g)\varphi(q')d\mu(q'),\quad\left.
\frac{d}{dt}(T_{\exp(tX)}^{\lambda}\varphi)\right|_{t=0}(q)=\ell_{X}(q,\lambda)\varphi(q),
\end{equation}
satisfy the relations
\begin{gather}\nonumber
	D_{q\overline{q}'}^{\lambda}(g_{1}\cdot g_{2})=\int_{Q}D_{q\overline{q}''}^{\lambda}(g_{1})D_{q''\overline{q}'}^{\lambda}(g_{2})d\mu(q''),\\  D_{q\overline{q}'}^{\lambda}(g)=\overline{D_{q'\overline{q}}^{\lambda}(g^{-1})},\quad D_{q\overline{q'}}^{\lambda}(e)=\delta(q,\overline{q'}),\label{condD}
\end{gather}
where $g_{1},g_{2}\in G$, and possess the properties of orthogonality and completeness \cite{barut}:
\begin{gather}
\int_{G}\overline{D_{\tilde{q}\overline{\tilde{q}'}}^{\tilde{\lambda}}}(g)D_{q\overline{q}'}^{\lambda}(g)d\mu(g)=\delta(q,\overline{\tilde{q}})\delta(\tilde{q}',\overline{q}')\delta(\tilde{\lambda},\lambda),\label{Dort}\\
\int_{Q\times Q\times J}\overline{D_{q\overline{q'}}^{\lambda}(\tilde{g})}D_{q\overline{q}'}^{\lambda}(g)d\mu(q)d\mu(q')d\mu(\lambda)=\delta(\tilde{g},g).\label{Dful}
\end{gather}

Note that the functions $D_{q\overline{q'}}^{\lambda}(g)$ are defined
globally on the Lie group $G$ iff the Kirillov condition of integerness
of the orbit $\mathcal{O}_{\lambda}$ holds \cite{kirr,spK}:
\begin{equation}
\frac{1}{2\pi}\int_{\gamma\in H_{1}(\mathcal{O}_{\lambda})}\omega_{\lambda}=n_{\gamma}\in\mathbb{Z},\label{cel_D}
\end{equation}
where $H_{1}(\mathcal{O}_{\lambda})$ is a one-dimensional homology
group of the stationarity group $G^{\lambda}=\{g\in G\mid Ad_{g}^{*}\lambda=\lambda\}$.
The functions $D_{q\overline{q'}}^{\lambda}(g)$ are eigenfunctions
for Casimir operators $\hat{K}_{\mu}(g)=K_{\mu}(-i\xi(g))=K_{\mu}(-i\eta(g))$
of the Lie group $G$:
\begin{gather}
\hat{K}_{\mu}(g)D_{q\overline{q'}}^{\lambda}(g)=\omega_{\mu}(\lambda)D_{q\overline{q'}}^{\lambda}(g).\label{KasD}
\end{gather}

Indeed, from (\ref{tD}) it follows that $\hat{K}_{\mu}(g)D_{q\overline{q'}}^{\lambda}(g)=K_{\mu}(i\ell(q,\lambda))D_{q\overline{q'}}^{\lambda}(g)$.
Then, in view of (\ref{lrpr}) and the homogeneity of Casimir functions, we obtain (\ref{KasD}).

\section{Noncommutative integration of the Klein--Gordon equation\label{sec:Noncommutative-integration}}

The Klein--Gordon equation for a complex scalar field $\tilde{\varphi}(g)$
on the space $\tilde{M}$ can be written as
\begin{gather}
\left(\tilde{\square}+\zeta\tilde{R}+m^{2}\right)\tilde{\varphi}(\tau,g)=0,\quad\zeta=\frac{n-2}{4(n-1)}.\label{kgtM}
\end{gather}
Here $m$ is a mass of the field $\tilde{\varphi}(\tau,g)$,
$\tilde{\square}$ is the d'Alembertian in $\tilde{M}$ of the form
\begin{gather*}
\tilde{\square}=a^{-2}(\tau)\left(\partial_{\tau}^{2}+(n-2)c(\tau)\partial_{\tau}-\Delta_{G}\right),
\end{gather*}
and $\Delta_{G}$ is the Laplace operator on the Lie group G. The
basis of solutions, $\{\tilde{\varphi}_{\sigma}(\tau,g)\} $, to equation (\ref{kgtM}) is sought in the form
\begin{gather}
\tilde{\varphi}_{\sigma}(\tau,g)=a^{\frac{2-n}{2}}(\tau)\varphi_{\sigma}(\tau,g),\quad\varphi_{\sigma}(\tau,g)=f_{\Lambda}(\tau)\Phi_{\sigma}(g),\quad\Lambda\in\sigma,\label{class_sol}
\end{gather}
where $\sigma$ is a set of quantum numbers, which parameterizes the basis solutions, and $\varphi_{\sigma}(\tau,g)$
describes a scalar field in the space $M$. Then the functions $\varphi_{\sigma}(\tau,g)$
and $\Phi_{\sigma}(g)$ satisfy the equations
\begin{equation}
\left(\partial_{\tau}^{2}-\Delta_{G}+\zeta R+a^{2}(\tau)m^{2}\right)\varphi(\tau,g)=0,\label{kg01}
\end{equation}
and
\begin{gather}
-\Delta_{G}\Phi_{\sigma}(g)=\Lambda^{2}\Phi_{\sigma}(g),\label{kf}
\end{gather}
respectively. The evolution of the scale factor $f_{\Lambda}(\tau)$,
dependent on the conformal time $\tau,$ is governed by the oscillator
equation with variable frequency:
\begin{gather}
\ddot{f}_{\Lambda}+\omega^{2}(\tau)f_{\Lambda}=0,\quad\omega^{2}(\tau)=\Lambda^{2}+a^{2}(\tau)m^{2}+\zeta R.\label{eq:ff3}
\end{gather}
The normalization condition
\begin{gather*}
\sqrt{\gamma}\int_{G}\overline{\Phi_{\sigma}(g)}\Phi_{\sigma'}(g)d\mu(g)=\delta(\sigma,\sigma'),
\end{gather*}
 is imposed on the solutions of equation (\ref{kf}). Here, $\gamma=\mathrm{det}(\gamma_{ab})$,
and $d\mu(g)$ is a Haar invariant measure on the Lie group $G$.

Note that the function $\dot{\overline{f_{\Lambda}}}f_{\Lambda}-\overline{f_{\Lambda}}\dot{f_{\Lambda}}$
is constant on solutions of equation (\ref{eq:ff3}). We also impose
the normalization condition

\begin{equation}
i\left(\dot{\overline{f_{\Lambda}}}f_{\Lambda}-\overline{f_{\Lambda}}\dot{f_{\Lambda}}\right)=1\text{.}\label{eq:normf}
\end{equation}
on solutions of this equation.

If the metric changes adiabatically (adiabatic approximation), $\dot{a}/a\ll1$,
then the solution of the equation (\ref{eq:ff3}) can be found in
the form of a generalized WKB approximation \cite{and87}:
\begin{equation}
f_{\Lambda}(\tau)=\frac{1}{\sqrt{2W(\tau)}}\exp\left(i\int_{\tau_{0}}^{\tau}W(\tau')d\tau'\right).\label{ffW}
\end{equation}
This function satisfies the Wronskian condition (\ref{eq:normf}).
For the function $W(\tau)$ we have the nonlinear equation
\begin{equation}
W^{2}(\tau)=\omega^{2}(\tau)-\frac{1}{2}\left(\frac{\ddot{W}(\tau)}{W(\tau)}-\frac{3}{2}\left(\frac{\dot{W}(\tau)}{W(\tau)}\right)^{2}\right).\label{myW}
\end{equation}
Equation \eqref{myW} can be solved iteratively as follows (see, e.g., \cite{Kohri2017}).
Let us start with the zero-order term, which contains no time derivatives,
i.e. $W_{(0)}(\tau)=\omega(\tau)$. Using $W_{(0)}(\tau)$ in the
right hand side of (\ref{myW}), we can find the
solution $W_{(2)}(\tau)$, which contains terms
with time derivatives up to the second order.

Substituting now $W_{(2)}(\tau)$ in the the right hand side of \eqref{myW}, we obtain
\begin{align}
W_{(4)}(\tau) & =\omega-\frac{1}{4}\frac{m^{2}a^{2}}{\omega^{3}}\left(\frac{\dot{a}^{2}}{a^{2}}+\frac{\ddot{a}}{a}\right)+\frac{1}{16}\frac{m^{2}a^{2}}{\omega^{5}}\left(3\frac{\ddot{a}^{2}}{a^{2}}+
4\frac{\dot{a} a^{(3)}}{a^2}+\frac{a^{(4)}}{a}\right)+\nonumber \\
 & +\frac{5}{8}\frac{m^{4}a^{4}}{\omega^{5}}\frac{\dot{a}^{2}}{a^{2}}-\frac{1}{32}\frac{m^{4}a^{4}}{\omega^{7}}\left(19\frac{\dot{a}^{4}}{a^{4}}+122\frac{\dot{a}^{2}}{a^{2}}\frac{\ddot{a}}{a}+19\frac{\ddot{a}^{2}}{a^{2}}+28\frac{\dot{a}{a^{(3)}}}{a^{2}}\right)+\nonumber \\
 &+\frac{221}{32}\frac{m^{6}a^{6}}{\omega^{9}}\left(\frac{\dot{a}^{4}}{a^{4}}+\frac{\ddot{a}\dot{a}^{2}}{a^{3}}\right)-\frac{1105}{128}\frac{m^{8}a^{8}}{\omega^{11}}\frac{\dot{a}^{4}}{a^{4}}\label{W4}
\end{align}
where $\dot{a}=da(\tau)/d\tau$, $\ddot{a}=d^2 a(\tau)/d\tau^2$, $a^{(k)}=d^k a(\tau)/d\tau^k$, $k>2 $. We can see that $W_{(4)}(\tau)$ contains time derivatives of the fourth order.

As the metric on the Lie group $G$ is bi-invariant, the Laplace operator $\Delta_{G}$ is the Casimir operator $H(-i\eta(g))$ on the Lie group $G$, where $H(f)=\gamma^{ab}f_{a}f_{b}$.

In Section \ref{sec:-representation}, we defined the functions $D_{q\overline{q'}}^{\lambda}(g^{-1})$. They are determined by operators of the $\lambda$-representation (\ref{lrpr}) of the Lie algebra $\mathfrak{g}$ from the system of equations (\ref{tD}), are eigenfunctions for Casimir operators (see (\ref{KasD})), and satisfy the completeness relations (\ref{Dful}). Therefore, it is convenient to choose the set of functions
\[
\Phi_{\sigma}(g)=\gamma^{-1/4}D_{q\overline{q'}}^{\lambda}(g^{-1}),\quad\sigma=(q,q',\lambda)
\]
 as the basis of solutions to the equation (\ref{kf}) with eigenvalues $\Lambda^{2}(\lambda)=H(-i\ell(0,\lambda))$.

Below the function $f_{\Lambda}(\tau)$ will be denoted as $f_{\lambda}(\tau)$.

\section{Vacuum expectation values of the energy-momentum tensor for a scalar field \label{sec:Vacuum-average}}

Let us perform the second quantization of a charged quantum field
$\varphi$. To this end, we expand the field $\varphi$ in terms of a complete system of solutions to the equation (\ref{kg01}):
\begin{align*}
\varphi(\tau,g) & =\int d\mu(\sigma)\left[\varphi_{\sigma}^{(+)}(\tau,g)a_{\sigma}^{(+)}+\varphi_{\sigma}^{(-)}(\tau,g)a_{\sigma}^{(-)}\right],\quad\varphi_{\sigma}^{(-)}(\tau,g)=\overline{\varphi_{\sigma}^{(+)}(\tau,g)}.
\end{align*}
Impose the canonical commutation relations
\[
\left[a_{\sigma}^{(-)},a_{\sigma'}^{(+)\dagger}\right]=\left[a_{\sigma}^{(-)\dagger},a_{\sigma'}^{(+)}\right]=\delta(\sigma,\sigma'),\quad\left[a_{\sigma}^{(\pm)},a_{\sigma'}^{(\pm)}\right]=\left[a_{\sigma}^{(\pm)\dagger},a_{\sigma'}^{(\pm)\dagger}\right]=0,
\]
where $a_{\sigma}^{(+)}$ and $a_{\sigma}^{(-)}$ are the creation operator of antiparticles and
 the annihilation operators of particles, respectively. The adjoint
operators $a_{\sigma}^{(-)\dagger}$ and $a_{\sigma}^{(+)\dagger}$
are the antiparticle annihilation operators and the particle creation
operators, respectively.

A vacuum state of the scalar field $\varphi $ is determined by the equations
\begin{equation}\label{vac00}
a_{\sigma}^{(-)}\mid0\rangle=a_{\sigma}^{(-)\dagger}\mid0\rangle=0,\quad\langle0\mid0\rangle=1.
\end{equation}
The EMT for the scalar field $\varphi(\tau,g)$ on the space-time $M$ reads \cite{Grib}:
\begin{gather}\nonumber
T(\eta_{X},\eta_{Y};m)\{\varphi,\varphi\}= \\ \nonumber
=\left(1-2\zeta\right)\overline{\eta_{(X}\varphi}\eta_{Y)}\varphi+\left(2\zeta-1/2\right)\mathbf{G}(X,Y)G^{AB}\overline{\eta_{A}\varphi}\eta_{B}\varphi-\\ \nonumber
-\left[\zeta\mathbf{R}(X,Y)+\left(2\zeta-1/2\right)\mathbf{G}(X,Y)(m^{2}+\zeta R)\right]\overline{\varphi}\varphi-\\
\zeta[(\nabla_{\eta_{(X}}\nabla_{\eta_{Y)}}\overline{\varphi})\varphi+\overline{\varphi}(\nabla_{\eta_{(X}}\nabla_{\eta_{Y)}})\varphi],\label{Teta01}
\end{gather}
where $\overline{\eta_{(X}\varphi}\eta_{Y)}\varphi = (\overline{\eta_{X}\varphi}\eta_{Y}\varphi + \overline{\eta_{Y}\varphi}\eta_{X}\varphi)/2$, $\nabla_{\eta_{X}}$ is the covariant derivative along a vector field $\eta_{X}$, $\nabla_{\eta_{(X}}\nabla_{\eta_{Y)}} = (\nabla_{\eta_{X}}\nabla_{\eta_{Y}}+\nabla_{\eta_{Y}}\nabla_{\eta_{X}})/2$, $X,Y\in\mathbb{R}\times\mathfrak{g}$.

The EMT for a scalar field on the space-time $\tilde{M}$ is shown to be related
to the initial EMT on $M$ as \cite{key-1}:
\begin{equation}
\tilde{T}(\eta_{X},\eta_{Y};m)\{\tilde{\varphi},\tilde{\varphi}\}=a^{2-n}(\tau)T(\eta_{X},\eta_{Y};a(\tau)m)\{\varphi,\varphi\}.\label{cod}
\end{equation}

Consider a vacuum expectation value of the EMT (\ref{cod}) on $M$ relative to
the vacuum state determined by equalities (\ref{vac00})
\begin{equation}\label{vacsr}
\langle T(\eta_{X},\eta_{Y};a(\tau)m)\rangle_{0}=\int T(\eta_{X},\eta_{Y};a(\tau)m)\{\varphi_{\sigma}(\tau,g),\varphi_{\sigma}(\tau,g)\}d\mu(\sigma).
\end{equation}
Here, the vacuum expectation values of the EMT on the space
$M$ are described by the integral over all quantum numbers of the EMT taken over the complete set of solutions to the Klein--Gordon equation (summation is understood over discrete quantum numbers).

In general, the vacuum averages of the product of fields on $M$ are defined in the same way:
\begin{equation}\nonumber
   \langle \hat{\varphi} \hat{\varphi} \rangle_0 =
    \int \overline{\varphi_{\sigma}(\tau,g)}\varphi_{\sigma}(\tau,g)d\mu(\sigma).
\end{equation}
Then, in view of the expression (\ref{cod}) for the vacuum expectation values of the EMT on the space $\tilde {M}$, relative to the vacuum state defined by the equations (\ref{vac00}), we obtain
\begin{gather*}
\langle\hat{\tilde{T}}(\eta_{X},\eta_{Y};m)\rangle_{0}=a^{2-n}(\tau)\langle T(\eta_{X},\eta_{Y};a(\tau)m)\rangle_{0}.
\end{gather*}

A bi-invariant metric on the Lie group $G$ leads to an important property of vacuum averages, which describes the following Theorem.
\begin{thr}
The vacuum expectation values $\langle\eta_{X}\hat{\varphi}\eta_{Y}\hat{\varphi}\rangle_{0}$
on a Lie group manifold $M$ with a bi-invariant metric have the $Ad_{g}$-invariance
property \label{Th01}
\begin{gather}
\langle\eta_{Ad_{g}X}\hat{\varphi}\eta_{Ad_{g}Y}\hat{\varphi}\rangle_{0}=\langle\eta_{X}\hat{\varphi}\eta_{Y}\hat{\varphi}\rangle_{0},\quad g\in G,\quad X,Y\in\mathfrak{g}.\label{futv1}
\end{gather}
\end{thr}
\begin{proof}
Consider the expression for vacuum expectation values
$\langle\eta_{X}\hat{\varphi}\eta_{Y}\hat{\varphi}\rangle_{0}$
taking into account (\ref{class_sol}):
\begin{gather}\nonumber
\langle\eta_{X}\hat{\varphi}\eta_{Y}\hat{\varphi}\rangle_{0}=\\ \nonumber
=\frac{1}{\sqrt{\gamma}}\int_{Q\times Q\times J}|f_{\lambda}(\tau)|^{2}\overline{\left(\eta_{X}D_{q\overline{q'}}^{\lambda}(g^{-1})\right)}\left(\eta_{Y}D_{q\overline{q'}}^{\lambda}(g^{-1})\right)d\mu(q)d\mu(q')d\mu(\lambda)=\\
=\frac{1}{\sqrt{\gamma}}\int_{Q\times Q\times J}|f_{\lambda}(\tau)|^{2}\left(\ell_{X}(q',\lambda)\overline{D_{q\overline{q'}}^{\lambda}(g^{-1})}\right)\left(\overline{\ell_{Y}(q',\lambda)}D_{q\overline{q'}}^{\lambda}(g^{-1})\right)\times \nonumber \\
\times d\mu(q)d\mu(q')d\mu(\lambda) \label{vac_03}
\end{gather}
Substituting the functions $D_{q\overline{q'}}^{\lambda}(g^{-1})$
in the form
\begin{gather*}
D_{q\overline{q'}}^{\lambda}(g^{-1})=\int_{Q}D_{q\overline{q''}}^{\lambda}(g^{-1})D_{q''\overline{q'}}^{\lambda}(e)d\mu(q'')
\end{gather*}
then, in view of relations (\ref{condD}), we obtain
\begin{gather}\nonumber
\langle\eta_{X}\hat{\varphi}\eta_{Y}\hat{\varphi}\rangle_{0}=\\
-\frac{1}{\sqrt{\gamma}}\int_{Q\times J}|f_{\lambda}(\tau)|^{2}
\left[\overline{\ell_{X}(q',\lambda)\ell_{Y}(q',\lambda)}D_{q\overline{q'}}^{\lambda}(e)\right]_{q=q'}
d\mu(q')d\mu(\lambda).\label{vv}
\end{gather}
By analogy with \eqref{vv}, we come to the following expression:
\begin{gather}
\langle\eta_{Ad_{g}X}\hat{\varphi}\eta_{Ad_{g}Y}\hat{\varphi}\rangle_{0}=
\langle\xi_{X}\hat{\varphi}\xi_{Y}\hat{\varphi}\rangle_{0}=\nonumber \\ \nonumber
=\frac{1}{\sqrt{\gamma}}\int_{Q\times Q\times J}|f_{\lambda}(\tau)|^{2}\overline{\left(\ell_{X}(q,\lambda)D_{q\overline{q'}}^{\lambda}(e)\right)}
\left(\ell_{Y}(q,\lambda)D_{q\overline{q'}}^{\lambda}(e)\right)\times \\ \times d\mu(q)d\mu(q')d\mu(\lambda).\label{eqq3}
\end{gather}
Summing the equations (\ref{tD}) and substituting in them $g=e$, we get the relation
\begin{equation}
\left(\ell_{X}(q,\lambda)+\overline{\ell_{X}(q',\lambda)}\right)D_{q\overline{q'}}^{\lambda}(e)=0.\label{deltaEQ}
\end{equation}
In view of (\ref{deltaEQ}), for (\ref{eqq3}), we have
\begin{gather}\nonumber
   \langle\eta_{Ad_{g}X}\hat{\varphi}\eta_{Ad_{g}Y}\hat{\varphi}\rangle_{0}=\\ \nonumber
   =\frac{1}{\sqrt{\gamma}}\int_{Q\times Q\times J}|f_{\lambda}(\tau)|^{2}\left(\ell_{X}(q',\lambda)\overline{D_{q\overline{q'}}^{\lambda}(e)}\right)
\left(\overline{\ell_{Y}(q',\lambda)}D_{q\overline{q'}}^{\lambda}(e)\right)d\mu(q)d\mu(q')d\mu(\lambda)=\\ \nonumber
= -\frac{1}{\sqrt{\gamma}}\int_{Q\times J}|f_{\lambda}(\tau)|^{2}
\left[\overline{\ell_{X}(q',\lambda)\ell_{Y}(q',\lambda)}D_{q\overline{q'}}^{\lambda}(e)\right]_{q=q'}
d\mu(q')d\mu(\lambda) = \\ \nonumber
=\langle\eta_{X}\hat{\varphi}\eta_{Y}\hat{\varphi}\rangle_{0}.
\end{gather}
\end{proof}

Equation (\ref{vac_03}) and unitarity of the $\lambda$-representation lead to the following lemma.

\begin{lemma}
The operators $\eta_{X}$ are skew-Hermitian with respect
to the vacuum expectation values:
\begin{equation}
\langle\eta_{X}\hat{\varphi}\eta_{Y}\hat{\varphi}\rangle_{0}=-\langle(\eta_{Y}\eta_{X}\hat{\varphi})\hat{\varphi}\rangle_{0}.\label{vac_eta}
\end{equation}
\end{lemma}
Note the important property of vacuum expectation values:
\begin{lemma}
The vacuum expectation values
 $\langle\hat{\varphi}\eta_{X}\varphi\rangle_{0}$
on a Lie group $M$ with a bi-invariant metric are zero.
\end{lemma}
\begin{proof}
	Carrying out the same analysis as in the proof of Theorem \ref{Th01}, we get
\begin{equation}
\langle\hat{\varphi}\eta_{X}\hat{\varphi}\rangle_{0}=-\langle\hat{\varphi}\xi_{X}\hat{\varphi}\rangle_{0}=\langle\hat{\varphi}\eta_{Ad_{g}X}\hat{\varphi}\rangle.\label{sv34}
\end{equation}
We will consider the vacuum expectation values $\langle\hat{\varphi}\eta_{a}\hat{\varphi}\rangle_{0}$
as a covector $f_{a}$. Then equality (\ref{sv34}) can be written
as

\begin{equation}
(Ad_{g}^{*}f)=f.\label{statf}
\end{equation}

The requirement (\ref{statf}) for a non-Abelian Lie algebra $\mathfrak{g}$
can be satisfied for all $g\in G$ only if $f=0$.
\end{proof}

Substituting expression (\ref{Teta01}) into (\ref{vacsr}) yields
\begin{gather*}
\langle\hat{\tilde{T}}(\eta_{0},\eta_{0})\rangle_{0}=\frac{1}{2}a^{2-n}(\tau)\left[\langle\hat{\dot{\varphi}}\hat{\dot{\varphi}}\rangle_{0}+\omega^{2}(\tau)\langle\hat{\varphi}\hat{\varphi}\rangle_{0}\right],\\
\langle\hat{\tilde{T}}(\eta_{0},\eta_{X})\rangle_{0}=\frac{1}{2}a^{2-n}(\tau)\left[\langle\hat{\dot{\varphi}}\eta_{X}\hat{\varphi}\rangle_{0}-\langle\hat{\varphi}\eta_{X}\hat{\dot{\varphi}}\rangle_{0}\right],\\
\langle\hat{\tilde{T}}(\eta_{X},\eta_{Y})\rangle_{0}=a^{2-n}(\tau)\bigg{\{}-\frac{1}{2}\langle\hat{\varphi}\{\eta_{X},\eta_{Y}\}\hat{\varphi}\rangle_{0}-\\
-\left(2\zeta-\frac{1}{2}\right)\gamma(X,Y)\left[\langle\hat{\dot{\varphi}}\hat{\dot{\varphi}}\rangle_{0}-\omega^{2}(\tau)\langle\hat{\varphi}\hat{\varphi}\rangle_{0}\right]\bigg{\}},
\end{gather*}
where property (\ref{vac_eta}) of the vacuum expectation value is taken into account.

From (\ref{vv}) we can find that
\begin{gather}
\langle\hat{\tilde{T}}(\eta_{0},\eta_{0})\rangle_{0}=\frac{1}{2\sqrt{\gamma}}a^{2-n}(\tau)\int_{J}\left(|f_{\lambda}(\tau)|^{2}\omega^{2}(\tau)+|\dot{f}_{\lambda}(\tau)|^{2}\right)\chi(\lambda)d\mu(\lambda), \label{T00} \\ \nonumber
\langle\hat{\tilde{T}}(\eta_{0},\eta_{X})\rangle_{0}=0,\\
\langle\hat{\tilde{T}}(\eta_{X},\eta_{Y})\rangle_{0}=-\frac{1}{\sqrt{\gamma}}a^{2-n}(\tau)\int_{J}(\left[F_{XY}(\lambda)+\zeta\mathbf{R}(X,Y)\chi(\lambda)\right]|f_{\lambda}(\tau)|^{2}+\label{Tab}\\
+\left(2\zeta-\frac{1}{2}\right)\gamma(X,Y)\left(|\dot{f}_{\lambda}(\tau)|^{2}-|f_{\lambda}(\tau)|^{2}\omega^{2}(\tau)\right)\chi(\lambda)),\nonumber
\end{gather}
where $\chi(\lambda)=\int_{Q}D_{q\overline{q}}^{\lambda}(e)d\mu(q)$
is the character of the $\lambda$-representation in the unit of
the Lie group $G,$
\begin{gather*}
F_{XY}(\lambda)=\frac{1}{2}\int_{Q}\left[\{\ell_{X}(q'),\ell_{Y}(q')\}_{+}D_{q\overline{q'}}^{\lambda}(e)\right]_{q'=q}d\mu(q).
\end{gather*}

\begin{lemma}
The symmetric 2-form $F_{XY}(\lambda)$ on $\mathfrak{g}$
is determined by the character of the $\lambda$-representation:
\begin{gather}
F_{XY}(\lambda)=-\frac{\Lambda_{\lambda}^{2}}{\mathrm{dim}G}\chi(\lambda)\gamma(X,Y).\label{condF2}
\end{gather}
\end{lemma}
\begin{proof}
The $Ad_{g}$- invariance property of the vacuum expectation values (\ref{futv1})
implies the $Ad_{g}-$ invariance of the symmetric 2-form $F_{XY}(\lambda)$:
\begin{gather}
F_{(Ad_{g}X)(Ad_{g}Y)}(\lambda)=F_{XY}(\lambda),\quad g\in G,\quad X,Y\in\mathfrak{g}.\label{AdF}
\end{gather}
A symmetric 2-form satisfying (\ref{AdF}) is known to define a bi-invariant
metric on a Lie group $G$. From the uniqueness of a bi-invariant
metric on a Lie group $G$ it follows that $F_{XY}(\lambda)=f(\lambda)\gamma(X,Y)$
where $f(\lambda)$ is a function of the parameter $\lambda$. Then, $\gamma^{ab}F_{ab}(\lambda)=\mathrm{dim}G\cdot f(\lambda)$. To find $f(\lambda)$, we find the convolution of $F_{XY}(\lambda)$
with the metric tensor $\gamma(X,Y)$:
\begin{gather*}
  \gamma^{ab}F_{ab}(\lambda)=-\int_{Q}H(-i\ell(q'))D_{q\overline{q'}}^{\lambda}(e)\mid_{q'=q}d\mu(q)=
  -\Lambda_{\lambda}^{2}\chi(\lambda) = \mathrm{dim}G\cdot f(\lambda).
\end{gather*}
Then we obtain (\ref{condF2}).
\end{proof}

Substituting (\ref{condF2}) into (\ref{Tab}) and taking into account
(\ref{RabeqR}), we get
\begin{gather}
\langle\hat{\tilde{T}}(\eta_{X},\eta_{Y})\rangle_{0}=\frac{\gamma(X,Y)}{n-1}\left(\langle\hat{\tilde{T}}(\eta_{0},\eta_{0})\rangle_{0}-a^{2}(\tau)\langle Sp(\hat{T})\rangle_{0}\right),\label{Tab2}
\end{gather}
where
\begin{equation}
\langle Sp(\hat{T})\rangle_{0}=G^{AB}\langle\hat{\tilde{T}}(\eta_{A},\eta_{B})\rangle_{0}=a^{2-n}(\tau)\frac{m^{2}}{\sqrt{\gamma}}\int_{J}|f_{\lambda}(\tau)|^{2}\chi(\lambda)d\mu(\lambda).\label{SpTunren}
\end{equation}
From (\ref{T00}) and (\ref{Tab2}) it follows that the vacuum expectation values of the EMT for a scalar field on the Lie group $\tilde{M}$ with a bi-invariant metric is determined by the character of the $\lambda$-representation
in the unit of the group.

Consider the adiabatic regularization of the vacuum expectation values
(\ref{Tab2})\textendash (\ref{SpTunren}) according to \cite{tarman,bunch,fulling,parker}.

To this end, we substitute expression \eqref{W4} for $W_{(4)}(\tau)$ into equation (\ref{Tab2}) and using (\ref{ffW}), we can obtain the following adiabatic vacuum contributions for the case of a commutative $ (n-1) $ -dimensional Lie group $ G = \mathbb {R}^{n-1} $ \cite{bunch}:
\begin{align*}
\langle\hat{\tilde{T}}(\eta_{0},\eta_{0})\rangle_{ad} & =
 \frac{1}{B_n}a^{2-n}(\tau)\int_0^\infty \lambda^{n-2}t_{0}[a,\lambda]d\lambda,\\
\langle Sp(\hat{T})\rangle_{ad} & =
\frac{1}{B_n}a^{-n}(\tau)\int_0^\infty \lambda^{n-2} t_{1}[a,\lambda]d\lambda,
\end{align*}

\begin{gather*}
	t_{0}[a,\lambda] =\omega+\frac{m^{4}a^{4}}{8\omega^{5}}\frac{\dot{a}^{2}}{a^{2}}-\frac{m^{4}a^{4}}{32\omega^{7}}\left(2\frac{{a}^{(3)}\dot{a}}{a^{2}}-\frac{\ddot{a}^{2}}{a^{2}}+4\frac{\ddot{a}\dot{a}^{2}}{a^{3}}-\frac{\dot{a}^{4}}{a^{4}}\right)+\\ \nonumber	+\frac{7m^{6}a^{6}}{16\omega^{9}}\left(\frac{\ddot{a}\dot{a}^{2}}{a^{3}}+\frac{\dot{a}^{4}}{a^{4}}\right)-\frac{105m^{8}a^{8}}{128\omega^{11}}\frac{\dot{a}^{4}}{a^{4}},
\end{gather*}
\begin{gather*}
	t_{1}[a,\lambda] =\frac{m^{2}a^{2}}{\omega}+\frac{m^{4}a^{4}}{4\omega^{5}}\left(\frac{\ddot{a}}{a}+\frac{\dot{a}^{2}}{a^{2}}\right)-\frac{5m^{6}a^{6}}{8\omega^{7}}\frac{\dot{a}^{2}}{a^{2}}-\\ \nonumber	-\frac{m^{4}a^{4}}{16\omega^{7}}\left(\frac{{a}^{(4)}}{a}+4\frac{{a}^{(3)}\dot{a}}{a^{2}}+3\frac{\ddot{a}^{2}}{a^{2}}\right)+\frac{m^{6}a^{6}}{32\omega^{9}}\left(28\frac{{a}^{(3)}\dot{a}}{a^{2}}+126\frac{\ddot{a}\dot{a^{2}}}{a^{3}}+  21\frac{\ddot{a}^{2}}{a^{2}}+21\frac{\dot{a}^{4}}{a^{4}}\right)-\\ \nonumber
  -\frac{231m^{8}a^{8}}{32\omega^{11}}\left(\frac{\ddot{a}\dot{a}^{2}}{a^{3}}+\frac{\dot{a}^{4}}{a^{4}}\right)+\frac{1155m^{10}a^{10}}{128\omega^{13}}\frac{\dot{a}^{4}}{a^{4}},
\end{gather*}
where $\omega^2 = \lambda^2 + a^2(\tau)m^2$, $B_n = 2^{n-1}\pi^{(n-1)/2}\Gamma((n-1)/2)$.

Within the framework of the adiabatic regularization, the renormalized final values of the vacuum
expectation value of the EMT are determined
by subtracting the adiabatic vacuum contributions from (\ref{T00})
and (\ref{SpTunren}), which do not depend on the global spatial topology \cite{and87}:
\begin{align}
\langle\hat{\tilde{T}}(\eta_{0},\eta_{0})\rangle_{ren} & =\langle\hat{\tilde{T}}(\eta_{0},\eta_{0})\rangle_{unren}-\frac{1}{B_n}a^{2-n}(\tau)\int_0^\infty \lambda^{n-2} t_{0}[a,\lambda]d\lambda,\nonumber \\
\langle Sp(\hat{T})\rangle_{ren} & =\langle Sp(\hat{T})\rangle_{unren}-\frac{1}{B_n}a^{-n}(\tau)\int_0^\infty \lambda^{n-2} t_{1}[a,\lambda]d\lambda,\label{Tren}
\end{align}
Note that the adiabatic regularization is not a regularization method for divergent integrals. Expressions (\ref{Tren}) consist of formally divergent integrals, and, in principle, to give them a mathematical
meaning, one needs to enter some covariant circumcision. There are two covariant methods that can be used to regularize vacuum expectation values: the covariant splitting of points and the dimensional regularization.
Then it is necessary to carry out a subtraction in (\ref{Tren}) and remove the regularization.

\section{Particle creation\label{sec:Particle-creation}}

Using the variables $\varphi$ and $\overline{\varphi}$ that satisfy
the equation of motion (\ref{kg01}), we construct the canonical Hamiltonian
of the scalar field:
\begin{equation}
	H(\tau)=\frac{1}{2}\sqrt{\gamma}\int_{G}\left\{ \overline{\dot{\varphi}}\dot{\varphi}+\gamma^{AB}\overline{\eta_{A}\varphi}\eta_{B}\varphi+\left(m^{2}a^{2}(\tau)+\zeta R\right)\overline{\varphi}\varphi\right\} d\mu(g).\label{mham}
\end{equation}

In order to diagonalize the Hamiltonian (\ref{mham}), one needs to expand
the field $\varphi$ into the complete basis of solutions ($\Phi_{J}(g)$)
to the equation (\ref{kg01}) with a set of quantum numbers $J$.
These functions satisfy the condition: there exists another set $\overline{J}$ of quantum numbers such that
\begin{equation}
\overline{\Phi_{J}(g)}=\theta_{J}\Phi_{\overline{J}}(g),\quad\left|\theta_{J}\right|=1.\label{condtheta}
\end{equation}
Note that the functions $\Phi_{\sigma}(g)$ do not possess the property
(\ref{condtheta}) in general, since the set of
quantum numbers $\sigma$ may contain the complex values of $q$ and $q'$.

We construct the necessary set of solutions as follows. Denote by
$A_{a}=A_{a}(f)$ a set of $\mathrm{dim}\mathcal{O}_{\lambda}/2$
independent functions in involution on the orbits of
$\mathcal{O}_{\lambda}$. Such a set for a compact semisimple Lie
group can be constructed by shifting the argument \cite{Bols} in the class of homogeneous polynomials.
The degree of homogeneity of a polynomial $A_{a}(f)$ is denoted
by $k_{a}\in\mathbb{N}$.

Denote by $\hat{A}_{a}=A_{a}(-i\xi)$ the maximal set of self-adjoint
pairwise commuting operators in the left-invariant enveloping algebra
of the Lie group $G$, and $\hat{B}_{a}=A_{a}(i\eta)$ is the maximum
set of self-adjoint pairwise commuting operators in a right-invariant
enveloping algebra of the Lie group $G$. The set of operators $\{\hat{A}_{a},\hat{B}_{b},\hat{K}_{\mu}\}$
of the dimension $\mathrm{dim}\mathfrak{g}$ forms a complete set of operators
on a Lie group $G$ \cite{barut}.
We will seek a basis
of solutions $\Phi_{J}(g)$ to equation (\ref{kf}) as a set of eigenfunctions
for a complete set of operators:
\begin{align}
\hat{A}_{a}\Phi_{J}(g) & =p_{a}\Phi_{J}(g),\nonumber \\
\hat{B}_{b}\Phi_{J}(g) & =s_{b}\Phi_{J}(g),\nonumber \\
\hat{K}_{\mu}\Phi_{J}(g) & =\omega_{\mu}(\lambda)\Phi_{J}(g),\label{sysps}
\end{align}
where $J=\{p,s,\lambda\}$ is a set of quantum numbers. The multindex
$p=\{p_{1},\dots,p_{\mathrm{dim}\mathcal{O}_{\lambda}/2}\}$ corresponds
to the set of eigenvalues of the $\hat{A}_{a}$, and $s=\{s_{1},\dots,s_{\mathrm{dim}\mathcal{O}_{\lambda}/2}\}$
corresponds to the set of eigenvalues of the $\hat{B}_{a}$ operators.
We impose the normalization condition
\[
\sqrt{\gamma}\int_{G}\overline{\Phi_{J}(g)}\Phi_{J'}(g)d\mu(g)=\delta(J,J')\text{.}
\]
We will seek a solution to the system (\ref{sysps}) in the form
\begin{equation}
\Phi_{J}(g)=\Phi_{ps}^{\lambda}(g)=\int\overline{\psi_{s}(q)}\psi_{p}(q')\Phi_{\sigma}(g)d\mu(q)d\mu(q').\label{rrpsi}
\end{equation}
Then from (\ref{tD}) we obtain
\begin{align*}
A_{a}(\ell(q',\lambda))\psi_{p}(q') & =p_{a}\psi_{p}(q'),\\
A_{a}(\ell(q,\lambda))\psi_{s}(q) & =s_{a}\psi_{s}(q).
\end{align*}
The functions $\psi_{p}(q')$ and $\psi_{s}(q)$ satisfy
the orthogonality and completeness relations
\begin{align}
	\nonumber
(\psi_{p},\psi_{p'}) & =\delta_{pp'},\quad(\psi_{s},\psi_{s'})=\delta_{ss'},\\
\sum_{p}\overline{\psi_{p}(\widetilde{q}')}\psi_{p}(q') & =\sum_{s}\overline{\psi_{s}(\widetilde{q}')}\psi_{s}(q')=\delta(\widetilde{q}'-q').\label{psifo}
\end{align}
Expression (\ref{rrpsi}) provides a correspondence between the basis
of solutions $\Phi_{J}(g)$ that are eigenfunctions for the complete
set of operators and the basis of solutions $\Phi_{\sigma}(g)$.

The complex conjugation of equality (\ref{rrpsi}) and the unitarity
of the $\lambda$-representation result in the expression
\begin{align*}
\overline{\Phi_{ps}^{\lambda}(g)} & =\gamma^{-1/4}\int\overline{\psi_{p}(q')}\psi_{s}(q)\overline{D_{q\overline{q}'}^{\lambda}(g^{-1})}d\mu(q)d\mu(q')=\Phi_{sp}^{\lambda}(g^{-1}).
\end{align*}
Whence it follows that $\overline{\Phi_{ps}^{\lambda}(g)}$ is also
an eigenfunction of the complete set of operators $\{\hat{A}_{a},\hat{B}_{b},\hat{K}_{\mu}\}$:
\begin{align*}
A_{a}(-i\xi(g))\overline{\Phi_{ps}^{\lambda}(g)} & =\left.A_{a}(-i\eta(\overline{g}))\Phi_{sp}^{\lambda}(\overline{g})\right|_{\overline{g}=g^{-1}}=\\
 & =(-1)^{k_{a}}\left.A_{a}(i\eta(\overline{g}))\Phi_{sp}^{\lambda}(\overline{g})\right|_{\overline{g}=g^{-1}}=\\
 & =(-1)^{k_{a}}p_{a}\overline{\Phi_{ps}^{\lambda}(g)}=\overline{p}_{a}\overline{\Phi_{ps}^{\lambda}(g)},
\end{align*}
\begin{align*}
A_{b}(i\xi(g))\overline{\Phi_{ps}^{\lambda}(g)} & =\left.A_{b}(i\xi(\overline{g}))\Phi_{sp}^{\lambda}(\overline{g})\right|_{\overline{g}=g^{-1}}=\\
 & =(-1)^{k_{b}}\left.A_{b}(-i\xi(\overline{g}))\Phi_{sp}^{\lambda}(\overline{g})\right|_{\overline{g}=g^{-1}}=\\
 & =(-1)^{k_{b}}s_{b}\overline{\Phi_{ps}^{\lambda}(g)}=\overline{s}_{b}\overline{\Phi_{ps}^{\lambda}(g)},
\end{align*}
\begin{align*}
K_{\mu}(-i\xi)\overline{\Phi_{ps}^{\lambda}(g)} & =\left.K_{\mu}(-i\eta(\overline{g}))\Phi_{sp}^{\lambda}(\overline{g})\right|_{\overline{g}=g^{-1}}=\\
 & =\omega_{\mu}(\lambda)\left.\Phi_{sp}^{\lambda}(\overline{g})\right|_{\overline{g}=g^{-1}}=\\
 & =\omega_{\mu}(\lambda)\overline{\Phi_{ps}^{\lambda}(g)}.
\end{align*}
Therefore, the function $\overline{\Phi_{ps}^{\lambda}(g)}$ is proportional
to the function $\Phi_{\overline{p}\overline{s}}^{\lambda}(g)$. The
proportionality factor is
equal to unit in modulus according
to the normalization condition.
In other words, the condition (\ref{condtheta}) is fulfilled, and
\[
J=\{p,s,\lambda\},\quad\overline{J}=\{\overline{p},\overline{s},\lambda\},\overline{p}_{a}=(-1)^{k_{a}}p_{a},\quad\overline{s}_{b}=(-1)^{k_{b}}s_{b}.
\]
 Let us expand the field $\varphi$ in terms of a complete system (\ref{rrpsi})
of solutions to equation (\ref{kg01}) numbered by the quantum numbers
$J$ and $\overline{J}$:
\begin{align*}
\varphi & =\int d\mu(J)\left[\varphi_{\overline{J}}^{(+)}a_{\overline{J}}^{(+)}+\varphi_{J}^{(-)}a_{J}^{(-)}\right]=\varphi^{(+)}+\varphi^{(-)},\\
\overline{\varphi} & =\int d\mu(J)\left[\varphi_{\overline{J}}^{(-)}a_{\overline{J}}^{(-)\dagger}+\varphi_{J}^{(+)}a_{J}^{(+)\dagger}\right],
\end{align*}
where $\varphi_{J}^{(+)}(\tau,g)=f_{\lambda}(\tau)\Phi_{J}(g),\quad\varphi_{J}^{(-)}(\tau,g)=\overline{\varphi_{J}^{(+)}(\tau,g)}.$
Impose the canonical commutation relations
\[
\left[a_{J}^{(-)},a_{J'}^{(+)\dagger}\right]=\left[a_{J}^{(-)\dagger},a_{J'}^{(+)}\right]=\delta(J,J'),\quad\left[a_{J}^{(\pm)},a_{J'}^{(\pm)}\right]=\left[a_{J}^{(\pm)\dagger},a_{J'}^{(\pm)\dagger}\right]=0.
\]
The Hamiltonian $H(\tau)$ in terms of the
creation and annihilation operators reads
\begin{gather}\nonumber
H(\tau) =\int d\mu(J)\bigg{\{} E_{J}(\tau)\left(a_{J}^{(+)\dagger}a_{J}^{(-)}+a_{\overline{J}}^{(-)\dagger}a_{\overline{J}}^{(+)}\right)+ \\ \nonumber
	+ F_{J}(\tau)a_{J}^{(+)\dagger}a_{\overline{J}}^{(+)}+\overline{F_{J}(\tau)}a_{\overline{J}}^{(-)\dagger}a_{J}^{(-)}\bigg{\}},\\ \label{Hbb}
E_{J}(\tau) =\left|\dot{f_{\lambda}}\right|+\omega^{2}\left|f_{\lambda}\right|,\quad F_{J}(\tau)=\overline{\theta_{J}}\left(\dot{f}_{\lambda}^{2}+\omega^{2}f_{\lambda}^{2}\right).
\end{gather}

Impose the initial conditions on the functions $f_{\lambda}(\tau)$:
\[
\dot{f}_{\lambda}(\tau_{0})=i\omega(\tau_{0})f_{\lambda}(\tau_{0}),\quad\left|f_{\lambda}(\tau_{0})\right|=(2\omega)^{-1/2}(\tau_{0}).
\]
Then $F_{J}(\tau_{0})=0$ and the Hamiltonian (\ref{Hbb})
is diagonal at the initial moment $\tau=\tau_{0}$ with respect to the
operators $a_{J}^{(\pm)}$ and $a_{J}^{(\pm)\dagger}$.

To diagonalize the Hamiltonian at an arbitrary instant $\tau$, we introduce
the operators $c_{J}^{(\pm)}$ and $c_{J}^{(\pm)\dagger}$ related
to the operators $a_{J}^{(\pm)}$ and $a_{J}^{(\pm)\dagger}$ by the
Bogolyubov canonical transformation:
\[
\begin{cases}
a_{J}^{(-)} & =\overline{\alpha_{J}(\tau)}c_{J}^{(-)}-\beta_{J}(\tau)\theta_{J}c_{\overline{J}}^{(+)},\\
a_{J}^{(+)} & =\alpha_{J}(\tau)c_{J}^{(+)}-\overline{\beta_{J}(\tau)\theta_{J}}c_{\overline{J}}^{(-)}.
\end{cases}
\]
For the adjoint operators, we have, respectively:
\[
\begin{cases}
a_{J}^{(-)\dagger} & =\overline{\alpha_{J}(\tau)}c_{J}^{(-)\dagger}-\beta_{J}(\tau)\theta_{J}c_{\overline{J}}^{(+)\dagger},\\
a_{J}^{(+)\dagger} & =\alpha_{J}(\tau)c_{J}^{(+)\dagger}-\overline{\beta_{J}(\tau)\theta_{J}}c_{\overline{J}}^{(-)\dagger}\text{,}
\end{cases}
\]
where the functions $\alpha_{J}(\tau)=\alpha_{\overline{J}}(\tau)$,
$\beta_{J}(\tau)=\beta_{\overline{J}}(\tau)$ satisfy the initial
conditions $\left|\alpha_{J}(\tau_{0})\right|=1$, $\beta_{J}(\tau_{0})=0$
and the relation $\left|\alpha_{J}(\tau)\right|^{2}-\left|\beta_{J}(\tau)\right|^{2}=1$.
The inverse transformations are
\begin{align}
\begin{cases}
c_{J}^{(-)} & =\alpha_{J}(\tau)a_{J}^{(-)}+\beta_{J}(\tau)\theta_{J}a_{\overline{J}}^{(+)},\\
a_{J}^{(+)} & =\overline{\alpha_{J}(\tau)}a_{J}^{(+)}+\overline{\beta_{J}(\tau)\theta_{J}}a_{\overline{J}}^{(-)},
\end{cases}\label{binv}\\ \nonumber
\begin{cases}
c_{J}^{(-)\dagger} & =\alpha_{J}(\tau)a_{J}^{(-)\dagger}+\beta_{J}(\tau)\theta_{J}a_{\overline{J}}^{(+)\dagger},\\
c_{J}^{(+)\dagger} & =\overline{\alpha_{J}(\tau)}a_{J}^{(+)\dagger}+\overline{\beta_{J}(\tau)\theta_{J}}a_{\overline{J}}^{(-)\dagger}\text{.}
\end{cases}
\end{align}
Then for the Hamiltonian, we get the expression
\begin{gather}\nonumber
	H(\tau) = \int d\mu(J)\bigg{\{} \widetilde{E}_{J}(\tau)\left(c_{J}^{(+)\dagger}c_{J}^{(-)}+c_{\overline{J}}^{(-)\dagger}c_{\overline{J}}^{(+)}\right)+\\ \nonumber
	+\widetilde{F}_{J}(\tau)c_{J}^{(+)\dagger}c_{\overline{J}}^{(+)}+\overline{\widetilde{F}_{J}(\tau)}c_{\overline{J}}^{(-)\dagger}c_{J}^{(-)}\bigg{\}},\\ \nonumber
\widetilde{E}_{J}(\tau) = E_{J}(\tau)\left(\left|\alpha_{J}(\tau)\right|^{2}+\left|\beta_{J}(\tau)\right|^{2}\right)-2\mathrm{Re}\left(F_{J}(\tau)\alpha_{J}(\tau)\overline{\beta_{J}(\tau)\theta_{J}}\right),\nonumber \\
\widetilde{F}_{J}(\tau) =-2\alpha_{J}(\tau)\beta_{J}(\tau)\theta_{J}E_{J}(\tau)+\alpha^{2}(\tau)F_{J}(\tau)+\beta_{J}^{2}(\tau)\theta_{J}^{2}\overline{F_{J}(\tau)}. \nonumber
\end{gather}
The condition $\widetilde{F}_{J}(\tau)=0$ of diagonalization of the
Hamiltonian at the moment $\tau$ with respect to the operators $\{c_{J}^{(\pm)}\text{, }c_{J}^{(\pm)\dagger}\}$
is compatible with the normalization condition (\ref{eq:normf}) only
if $\omega^{2}(\tau)>0$. This is equivalent to the requirement $k<0.$

From equation $\widetilde{F}_{J}(\tau)=0$, we obtain
\begin{gather}\nonumber
\alpha_{J}(\tau)=\frac{i}{\sqrt{2\omega(\tau)}}\chi_{J}(\tau)\left(\overline{\dot{f}_{\Lambda}(\tau)}-i\omega(\tau)\overline{f_{\Lambda}(\tau)}\right),\\ \nonumber
\quad\beta_{J}(\tau)=\frac{i}{\sqrt{2\omega(\tau)}}\chi_{J}(\tau)\left(\dot{f}_{\Lambda}(\tau)-i\omega(\tau)f_{\Lambda}(\tau)\right),
\end{gather}
where $\chi_{J}(\tau)=\chi_{\overline{J}}(\tau)$ is an arbitrary
complex function such that $\left|\chi_{J}(\tau)\right|=1$. It is convenient to modify the operators $\{c_{J}^{(\pm)}\text{, }c_{J}^{(\pm)\dagger}\}$ by
\[
d_{J}^{(+)}=\chi_{J}(\tau)c_{J}^{(+)},\quad d_{J}^{(-)}=\overline{\chi_{J}(\tau)}c_{J}^{(-)},
\]
where the operators $\{d_{J}^{(+)},d_{J}^{(-)}\}$ satisfy the same commutation relations as the original operators $\{c_{J}^{(\pm)}\text{, }c_{J}^{(\pm)\dagger}\}$.

Then the Hamiltonian $H(\eta)$ is diagonal with respect to the operators $\{d_{J}^{(+)},d_{J}^{(-)}\}$:
\[
H(\eta)=\int d\mu(J)\omega(\eta)\left(d_{J}^{(+)\dagger}d_{J}^{(-)}+d_{J}^{(-)\dagger}d_{J}^{(+)}\right).
\]
Suppose the quantized scalar field at the initial moment $\eta_{0}$
is in the state $\mid0\rangle$, which is annihilated by the operators
$\{a_{J}^{(\pm)}\text{, }a_{J}^{(\pm)\dagger}\}$. At the moment $\tau>\tau_{0}$
the vacuum state is defined as follows:
\[
d_{J}^{(-)}\mid0_{\tau}\rangle=d_{J}^{(-)\dagger}\mid0_{\tau}\rangle=0.
\]
In the Heisenberg picture, the state $\mid0\rangle$ is not vacuum one
subject to $\tau>\tau_{0}$. Using the inverse transformations (\ref{binv}),
we can easily find that in each mode $J$ this state contains $n_{J}(\tau)$
pairs of quasiparticles with quantum numbers $J$ and $\overline{J}$, where
\[
n_{J}(\tau)=\langle0\mid d_{J}^{(+)\dagger}d_{J}^{(-)}\mid0\rangle=\langle0\mid d_{J}^{(+)}d_{J}^{(-)\dagger}\mid0\rangle=\left|\beta_{J}(\tau)\right|^{2}.
\]
 The density of created particles is defined as the vacuum expectation values
relative to the instantaneous vacuum $\mid0_{\tau}\rangle$ of the
particle density operator,
\begin{gather}\nonumber
n(\tau,g)=\langle 0_{\tau}\mid\hat{n}(\tau,g)\mid0_{\tau}\rangle,\\ \nonumber
n(\tau,g)=-i\sqrt{\gamma}a^{1-n}(\tau)\left((\partial_{\tau}\varphi^{(+)})^{\dagger}\varphi^{(-)}-\varphi^{(+)\dagger}\partial_{\tau}\varphi^{(-)}\right),
\end{gather}
where $\varphi=\varphi^{(+)}+\varphi^{(-)}$ is expansion of the field
operator into positive- and negative-frequency parts. Then we have
\begin{gather*}
	n(\tau,g) = -i\sqrt{\gamma}a^{1-n}(\tau)\int d\mu(J)d\mu(J')\bigg{\{} (\dot{f}_{\lambda}\overline{f}_{\lambda'}-f_{\lambda}\dot{\overline{f}}_{\lambda'})\times \\
	\times \overline{\Phi_{J'}(g)}\Phi_{J}(g)\langle0_{\tau}\mid a_{J}^{(+)\dagger}a_{J'}^{(-)}\mid0_{\tau}\rangle\bigg{\}} =\\
  = \sqrt{\gamma}a^{1-n}(\tau)\int\left|\beta_{\lambda}(\tau)\Phi_{J}(g)\right|^{2}d\mu(J).
\end{gather*}
Simplifying this expression with the use of expansions (\ref{rrpsi})
and relations (\ref{psifo}) gives
\begin{gather*}
\sqrt{\gamma}\int\left|\beta_{\lambda}(\tau)\Phi_{J}(g)\right|^{2}d\mu(J)  = \\ \nonumber
=\sqrt{\gamma}\sum_{s,p}\int\left|\beta_{\lambda}(\tau)\right|^{2}\left(\int\psi_{s}(q)\overline{\psi_{p}(q')}\overline{\Phi_{\sigma}(g)}d\mu(q)d\mu(q')\right)\times \\ \nonumber
\times\left(\int\overline{\psi_{s}(\widetilde{q})}\psi_{p}(\widetilde{q}')\Phi_{\widetilde{\sigma}}(g)d\mu(\widetilde{q})d\mu(\widetilde{q}')\right)_{\widetilde{\lambda}=\lambda}d\mu(\lambda)=\\ \nonumber
= \sqrt{\gamma}\int\left|\beta_{\lambda}(\tau)\Phi_{\sigma}(g)\right|^{2}d\mu(q)d\mu(q')d\mu(\lambda)
= \int\left|\beta_{\lambda}(\tau)\right|^{2}\chi(\lambda)d\mu(\lambda).
\end{gather*}
Thus, the density of created pairs does not depend on the group coordinates.
It is determined by the scale factor $a(\tau)$ and the character of
the $\lambda$-representation in the identity element of the group $G$:
\begin{equation}
n(\tau)=a^{1-n}(\tau)\int\left|\beta_{\lambda}(\tau)\right|^{2}\chi(\lambda)d\mu(\lambda)\text{.}\label{ntt}
\end{equation}
The expression (\ref {ntt}) depends on the topology of the Lie group $G$, and the integral over the quantum numbers $\lambda$ is determined only by integer orbits.

\section{Vacuum expectation values of the energy-momentum tensor for a scalar field on
$\mathbb{R}\times SO(3)$\label{sec:SO3}}

Consider the three-dimensional rotation group $SO(3)$ as the Lie
group $G$. Fix some basis $\{e_{a}\}$ of the Lie algebra $\mathfrak{so}(3)$:
\[
   [e_1,e_2] = e_3,\quad [e_3,e_1] = e_2,\quad [e_2,e_3] = e_1.
\]
A bi-invariant metric on $SO(3)$ is given by the 2-form $\gamma_{ab}=-\mathrm{2k\,diag}(1,1,1)$, \, $a,b=1,2,3$.
Without loss of generality, we set $k=-1/2$.

let us define a metric on the space-time $\tilde{M}$ in local coordinates as
\[
ds^{2}=a^{2}(\eta)\left(d\eta^{2}-dl^{2}\right),\quad dl^{2}=\left(d\phi^{2}+d\theta^{2}+2cos(\theta)d\phi d\psi+d\psi^{2}\right),
\]
where $g=(\phi,\theta,\psi)\in G$ are Euler angles, $\phi\in[0;2\pi),\quad\theta\in[0;\pi),\quad\psi\in[0;2\pi)$.
The bi-invariant metric $dl^{2}$ on the Lie group $SO(3)$ is a metric
of a three-dimensional sphere of radius $r=1$ and thus it coincides
with the metric of the closed Friedmann cosmological model.

Unlike the Friedmann model, the space in this case has
the topology of the projective space $\rm{PR}_{3}$ \cite{barut}. The
Ricci tensor and the scalar curvature of the manifold $M=\mathbb{R}\times SO(3)$
have the form:
\[
R_{ab}=\mathrm{diag}\left(0,-\frac{1}{2},-\frac{1}{2},-\frac{1}{2}\right),\quad R=\frac{3}{2}.
\]
The left-invariant vector field $\xi_{a}$ and the right-invariant vector field $\eta_{a}$
on the group $SO(3)$ in Euler angles are
\begin{gather*}
\xi_{1}=\frac{\sin\psi}{\sin\theta}\frac{\partial}{\partial\phi}+\cos\psi\frac{\partial}{\partial\theta}-\cot\theta\sin\psi\frac{\partial}{\partial\psi},\\
\xi_{2}=\frac{\cos\psi}{\sin\theta}\frac{\partial}{\partial\phi}-\sin\psi\frac{\partial}{\partial\theta}-\cot\theta\cos\psi\frac{\partial}{\partial\psi},\quad\xi_{3}=\frac{\partial}{\partial\psi},\\
\eta_{1}=\cot\theta\sin\phi\frac{\partial}{\partial\phi}-\cos\phi\frac{\partial}{\partial\theta}-\frac{\sin\phi}{\sin\theta}\frac{\partial}{\partial\psi},\\
\eta_{2}=-\cot\theta\cos\phi\frac{\partial}{\partial\phi}-\sin\phi\frac{\partial}{\partial\theta}+\frac{\cos\phi}{\sin\theta}\frac{\partial}{\partial\psi},\quad\eta_{3}=-\frac{\partial}{\partial\phi}.
\end{gather*}
Each non-degenerate integer coajoint orbit of the group $SO(3)$ passes through the covector $\lambda(j)=(j,0,0)$, where $j=1,2,3,\dots$ and the orbit is a two-dimensional sphere of radius $j^{2}$ centered at $(0,0,0)$:
\[
\mathcal{O_{\lambda}}=\left\{ f\in\mathbb{R}^{3}\left|K(f)=j^{2},f\neq0\right.\right\} ,
\]
where $K(f)=f_{1}^{2}+f_{2}^{2}+f_{3}^{2}$ is Casimir
function. The set of operators $\{\hat{A}=-i\xi_{3},\hat{B}=i\eta_{3},\hat{K}\}$,
forms a complete set of operators on $SO(3)$. The solution of the system
(\ref{sysps}) reads
\[
\Phi_{J}(g)=\frac{1}{\sqrt{8\pi^2}}D_{mn}^{j}(g),\quad J=\{n,m,j\},
\]
where $D_{mn}^{j}(g)$ is the Wigner D-matrix of $SO(3)$ \cite{barut}:
\begin{align}
D_{mn}^{j}(g) & =e^{im\phi+in\psi}d_{mn}^{j}(\theta),\label{dmn1}
\end{align}
\[
d_{mn}^{j}(\theta)=(-1)^{m-n}\sqrt{\frac{(j+m)!(j-m)!}{(j+n)!(j-n)!}}\sin^{m-n}\frac{\theta}{2}\cos^{m+n}\frac{\theta}{2}P_{j-m}^{(m-n,m+n)}(\cos\theta).
\]
Here $P_{n}^{(\alpha,\beta)}(z)$ are the Jacobi polynomials,
\[
P_{n}^{(\alpha,\beta)}(z)=\frac{(-1)^{n}}{2^{n}n!}(1-z)^{-\alpha}(1+z)^{-\beta}\frac{d^{n}}{dz^{n}}\left[(1-z)^{n+\alpha}(1+z)^{n+\beta}\right]\text{.}
\]
The Wigner functions satisfy the orthogonality and completeness conditions
\begin{gather*}
\frac{1}{8\pi^2}
\int_{G}\overline{D_{mn}^{j}(g)}D_{\tilde{m}\tilde{n}}^{\tilde{j}}(g)d\mu(g)=\frac{\delta_{j\tilde{j}}}{2j+1}\delta_{m\tilde{m}}\delta_{n\tilde{n}},\\
\sum_{n=-j}^{j}\overline{D_{mn}^{j}(g)}D_{\tilde{m}n}^{j}=\delta_{m\tilde{m}},
\end{gather*}
where $d\mu (g)$ is the Haar measure defined by the formula
\begin{equation}\nonumber
\int_{G}(\cdot)d\mu(g)=\int_{0}^{2\pi}d\psi\int_{0}^{\pi}\sin\theta d\theta\int_{0}^{2\pi}d\phi(\cdot).
\end{equation}
From (\ref{dmn1}) it follows
\[
\theta_{J}=(-1)^{m-n},\quad\overline{J}=\{-m,-n,j\}.
\]
The complex polarization $\mathfrak{p}=\{e_{1},e_{2}+ie_{3}\}$ of the covector $\lambda(j)$ corresponds to the operators of $\lambda$-representation (\ref{lrpr}) \cite{br2011}
\begin{gather*}
\ell_{1}(q,j)=-i(\sin(q)\frac{\partial}{\partial q}-j\cos(q)),\\ \nonumber
\ell_{2}(q,j)=-i(\cos(q)\frac{\partial}{\partial q}+j\sin(q)),\quad\ell_{3}(q,j)=\frac{\partial}{\partial q}.\label{lprSO3}
\end{gather*}
The operators $-i\ell_{X}(q,j)$ are Hermitian with respect to the
scalar product
\[
(\psi_{1},\psi_{2})=\int_{Q}\overline{\psi_{1}(q)}\psi_{2}(q)d\mu(q),\quad d\mu(q)=\frac{(2j+1)!}{2^{j}(j!)^{2}}\frac{dq\wedge d\overline{q}}{(1+\cos(q-\overline{q}))^{j+1}}.
\]
The functions $D_{q\overline{q'}}^{\lambda}(g^{-1})$ are found by integrating the system of equations (\ref{tD}) and have the form
\begin{gather}\nonumber
D_{q\overline{q'}}^{\lambda}(g^{-1})=D_{q\overline{q}'}^{j}(\phi,\theta,\psi)=
\frac{1}{\sqrt{8\pi^2}}\frac{2^{j}(j!)^{2}}{(2j)!}\times \\
\times\left[\cos\theta+\cos(\phi+\overline{q}')\cos(\phi-q)e^{-i\theta}+\sin(\phi+\overline{q}')\sin(\psi-q)\right]^{j}.\label{Dso3}
\end{gather}
The functions (\ref{Dso3}) satisfy the orthogonality and completeness conditions
(\ref{Dort})\textendash (\ref{Dful}) with respect
to the measure $d\mu(\lambda)$ and the delta-functions $\delta(q,\overline{q'})$:
\[
\int_{J}(\cdot)d\mu(\lambda)=\sum_{j=0}^{\infty}(2j+1)(\cdot),\quad\delta(q,\overline{q'})=\frac{2^{j}(j!)^{2}}{(2j)!}\left(1+\cos(q-\overline{q'})\right).
\]
In our case, the relationship (\ref{rrpsi}) takes the form
\begin{eqnarray*}
D_{mn}^{j}(g) & = & C_{mn}^{j}\int e^{i(n\overline{q}-mq')}D_{q\overline{q'}}^{\lambda}(g^{-1})d\mu(q)d\mu(q'),\\
C_{mn}^{j} & = & \sqrt{8\pi^2}e^{i\pi(n-m)/2}(j!)^{2}\left((j-m)!(j+m)!(j-n)!(j+n)!\right)^{-1/2}.
\end{eqnarray*}
Using expression (\ref{Dso3}) for the character of the $\lambda$-representation,
we get
\[
  \chi(\lambda)=\frac{1}{4\pi^2}\left(j+\frac{1}{2}\right).
\]

Thus, the Wigner D-matrix $D_{mn}^{j}(g)$ defines a basis
of solutions satisfying the condition (\ref{condtheta}). The density
of the created particles is given by
\begin{equation}\nonumber
n(\tau,g)=\frac{1}{4\pi^2}a^{1-n}(\tau)\sum_{j=0}^{\infty}\left(j+\frac{1}{2}\right)^{2}\frac{1}{\omega_{j}(\tau)}\left|\dot{f}_{j}(\tau)-i\omega_{j}(\tau)f_{j}(\tau)\right|^{2}.
\end{equation}

The vacuum expectation value of the EMT for the scalar field can
be written in the form
\begin{gather}
\langle\hat{\tilde{T}}(\eta_{0},\eta_{0})\rangle_{unren}=\frac{1}{4\pi^2 a^{2}(\tau)}\sum_{j=0}^{\infty}\left(j+\frac{1}{2}\right)^{2}\left(|f_{j}(\tau)|^{2}\omega_{j}^{2}(\tau)+|\dot{f}_{j}(\tau)|^{2}\right),\label{tso1}\\
\langle Sp\hat{\tilde{T}}\rangle_{unren}=\frac{m^{2}}{2\pi^2 a^{2}(\tau)}\sum_{j=0}^{\infty}\left(j+\frac{1}{2}\right)^{2}|f_{j}(\tau)|^{2},\nonumber
\end{gather}
where $\omega_{j}^{2}(\tau)=(j+1/2)^{2}+a^{2}(\tau)m^{2}$.
Expressions (\ref{tso1}) allow us to calculate the vacuum polarization
effect of the scalar field for a given nonstationary metric in the
space $\tilde{M}$.

Consider the vacuum expectation values of the EMT in the adiabatic approximation accurate to the fourth order:
\begin{align*}
\langle\hat{\tilde{T}}(\eta_{0},\eta_{0})\rangle_{(4)} & =\frac{1}{4\pi^2 a^{2}(\tau)}\sum_{j=0}^{\infty}\left(j+\frac{1}{2}\right)^{2}t_{0}[a,j],\\
\langle Sp(\hat{T})\rangle_{(4)} & =\frac{1}{2\pi^2 a^{2}(\tau)}\sum_{j=0}^{\infty}\left(j+\frac{1}{2}\right)^{2}t_{1}[a,j],
\end{align*}

\begin{gather*}
t_{0}[a,j]=\omega_{j}+\frac{m^{4}a^{4}}{8\omega_{j}^{5}}\frac{\dot{a}^{2}}{a^{2}}-\frac{m^{4}a^{4}}{32\omega_{j}^{7}}\left(2\frac{{a}^{(3)}\dot{a}}{a^{2}}-\frac{\ddot{a}^{2}}{a^{2}}+4\frac{\ddot{a}\dot{a}^{2}}{a^{3}}-\frac{\dot{a}^{4}}{a^{4}}\right)+\\ \nonumber
+\frac{7m^{6}a^{6}}{16\omega_{j}^{9}}\left(\frac{\ddot{a}\dot{a}^{2}}{a^{3}}+\frac{\dot{a}^{4}}{a^{4}}\right)-\frac{105m^{8}a^{8}}{128\omega_{j}^{11}}\frac{\dot{a}^{4}}{a^{4}},
\end{gather*}
\begin{gather*}
t_{1}[a,j]  =\frac{m^{2}a^{2}}{\omega_{j}}+\frac{m^{4}a^{4}}{4\omega_{j}^{5}}\left(\frac{\ddot{a}}{a}+\frac{\dot{a}^{2}}{a^{2}}\right)-\frac{5m^{6}a^{6}}{8\omega_{j}^{7}}\frac{\dot{a}^{2}}{a^{2}}-\\
-\frac{m^{4}a^{4}}{16\omega_{j}^{7}}\left(\frac{{a}^{(4)}}{a}+4\frac{{a}^{(3)}\dot{a}}{a^{2}}+3\frac{\ddot{a}^{2}}{a^{2}}\right)+\frac{m^{6}a^{6}}{32\omega_{j}^{9}}\left(28\frac{{a}^{(3)}\dot{a}}{a^{2}}+126\frac{\ddot{a}\dot{a^{2}}}{a^{3}}+21\frac{\ddot{a}^{2}}{a^{2}}+21\frac{\dot{a}^{4}}{a^{4}}\right)-\\
-\frac{231m^{8}a^{8}}{32\omega_{j}^{11}}\left(\frac{\ddot{a}\dot{a}^{2}}{a^{3}}+\frac{\dot{a}^{4}}{a^{4}}\right)+\frac{1155m^{10}a^{10}}{128\omega_{j}^{13}}\frac{\dot{a}^{4}}{a^{4}}.
\end{gather*}

To calculate the renormalized values in the adiabatic approximation,
we subtract the expressions (see \eqref{Tren})
\begin{align*}
\langle\hat{\tilde{T}}(\eta_{0},\eta_{0})\rangle_{ad} & =\frac{1}{4\pi^2 a^{2}(\tau)}\int_{0}^{\infty} j^{2}t_{0}[a,j] dj,\\
\langle Sp(\hat{T})\rangle_{ad} & =\frac{1}{2\pi^2 a^{2}(\tau)}\int_{0}^{\infty} j^{2} t_{1}[a,j]dj.
\end{align*}
Calculation of renormalized expressions reduces to obtaining expressions
of the form
\begin{equation}\label{defSp}
S(p)=\left\{ \sum_{j=1/2}^{\infty}-\int_{0}^{\infty}dj\right\} j^{2}\left(j^{2}+a^{2}(\tau)m^{2}\right)^{p+1/2}.
\end{equation}
To calculate sums over $j$, we use the Abel--Plan formula
\[
\sum_{j=0}^{\infty}F\left(j+\frac{1}{2}\right)=\int_{0}^{\infty}F(x)dx-i\int_{0}^{\infty}\frac{F(it)-F(-it)}{e^{2\pi t}+1}dt.
\]
Then we have
\begin{gather*}
\sum_{j=0}^{\infty}\left(j+\frac{1}{2}\right)^{2}\left(\left(j+\frac{1}{2}\right)^{2}+C^{2}\right)^{p+1/2} =\\
= \int_{0}^{\infty}t^{2}\left(t^{2}+C^{2}\right)^{p+1/2}dt+(-1)^{p}\Phi(p,C),\\
\Phi(p,C) = \int_{C}^{\infty}\frac{t^{2}\left(t^{2}-C^{2}\right)^{p+1/2}}{e^{2\pi t}+1}dt,\quad C = a(\tau)m.
\end{gather*}
Here we take into account the signs related to rounding the branch points $t=\pm iC$ of the
function $\sqrt{t^{2}+C^{2}}\left(t^{2}+C^{2}\right)^{p}$ using
the equality ( see, e.g., \cite{most01})
\[
F(it)-F(-it)=2i\,t^{2}\left(t^{2}-C^{2}\right)^{p+1/2}.
\]
As a result, for the sums
$S(p)$ we get the expression:
\[
S(p)=(-1)^{p}\Phi(p,C)\text{.}
\]
The renormalized vacuum expectation values % averages
in the adiabatic approximation have
the form
\begin{gather}\nonumber
	\langle\hat{\tilde{T}}(\eta_{0},\eta_{0})\rangle_{ren} =\frac{1}{4\pi^2 a^{2}(\tau)}\bigg{\{} S(0)+\frac{m^{4}a^{4}}{8S(-3)}\frac{\dot{a}^{2}}{a^{2}} -\\ \nonumber
	-\frac{m^{4}a^{4}}{32S(-4)}\left(2\frac{{a}^{(3)}\dot{a}}{a^{2}}-\frac{\ddot{a}^{2}}{a^{2}}+4\frac{\ddot{a}\dot{a}^{2}}{a^{3}}-\frac{\dot{a}^{4}}{a^{4}}\right)+ \\ \nonumber
	+\frac{7m^{6}a^{6}}{16S(-5)}\left(\frac{\ddot{a}\dot{a}^{2}}{a^{3}}+\frac{\dot{a}^{4}}{a^{4}}\right)-\frac{105m^{8}a^{8}}{128S(-6)}\frac{\dot{a}^{4}}{a^{4}}\bigg{\}},\nonumber \\
	\langle Sp(\hat{T})\rangle_{ren} = \frac{1}{2\pi^2 a^{2}(\tau)}\bigg{\{} \frac{m^{2}a^{2}}{S(-1)}+\frac{m^{4}a^{4}}{4S(-3)}\left(\frac{\ddot{a}}{a}+\frac{\dot{a}^{2}}{a^{2}}\right)-\nonumber \\
	-\frac{5m^{6}a^{6}}{8S(-3)}\frac{\dot{a}^{2}}{a^{2}}-\frac{m^{4}a^{4}}{16S(-4)} +
        \left(\frac{{a}^{(4)}}{a}+4\frac{{a}^{(3)}\dot{a}}{a^{2}}+3\frac{\ddot{a}^{2}}{a^{2}}\right)+\nonumber \\ \nonumber
	+\frac{m^{6}a^{6}}{32S(-5)}\left(28\frac{{a}^{(3)}\dot{a}}{a^{2}}+126\frac{\ddot{a}\dot{a^{2}}}{a^{3}}+21\frac{\ddot{a}^{2}}{a^{2}}+21\frac{\dot{a}^{4}}{a^{4}}\right)-\\
	-\frac{231m^{8}a^{8}}{32S(-6)}\left(\frac{\ddot{a}\dot{a}^{2}}{a^{3}}+\frac{\dot{a}^{4}}{a^{4}}\right)+\frac{1155m^{10}a^{10}}{128S(-7)}\frac{\dot{a}^{4}}{a^{4}}\bigg{\}}, \label{So3ren}
\end{gather}
Note that the four-dimensional closed homogeneous isotropic Robertson--Walker space $ K_4 $ of positive curvature $ R = 1 $ is described by the group $ G = SU (2) $. The group $ SU (2) $ is a universal covering group of $ SO (3) $ and it has the topology of a three-dimensional sphere. The vacuum expectation values of EMT in the adiabatic approximation are described by the same expressions (\ref{So3ren}), but with the function $S(p)$ equal to (see Ref. \cite{and87}):
\begin{equation}\label{defSp2}
   S(p)|_{K_4} = \left\{ \sum_{j=1}^{\infty}-\int_{0}^{\infty}dj\right\} j^{2}\left(j^{2}+a^{2}(\tau)m^{2}\right)^{p+1/2}.
\end{equation}
The difference between (\ref{defSp}) and (\ref{defSp2}) is due to the different topology of the Lie groups $SO(3)$ and $SU(2)$.

\section{Conclusion\label{sec:Conclusion}}

The vacuum expectation values of the EMT for a scalar field are shown to be invariant with respect to the
adjoint representation of a group Lie $G$ (Theorem 1). The expectation values are determined by the characters of the $\lambda$-representation of the Lie group and solutions of the oscillator equation with a variable frequency. The solutions of this equation are determined by the time dependence in the metric.

The tetrad components (\ref{T00}) and (\ref{Tab2}) for the non-renormalized vacuum expectation values of the EMT are found.

To obtain finite values of the EMT, we apply the adiabatic regularization method which is widely used in homogeneous and isotropic spaces \cite{parker}.
Calculations of the regularized terms corresponding to divergences in the EMT, $\langle\hat{\tilde T}(\eta_{X}, \eta_{Y}) \rangle_{unren}$, are to be carried out without taking into account the global spatial topology.

We also specified the procedure for constructing a basis of solutions
for which the property (\ref{condtheta}) holds. This is necessary
for the Hamiltonian diagonalization of the scalar field. An expression for
the density of created particles is found (\ref{ntt}).
The group topology in the effects of vacuum polarization and
particle creation is manifested under the orbital integrity condition (\ref{cel_D}).

The results obtained are illustrated by the example of the rotation group $SO(3)$. Such a space-time differs from the closed Robertson--Walker universe in its spatial topology. The space of the closed Robertson--Walker universe has the topology of a three-dimensional sphere, while the rotations group has the topology of the three-dimensional projective space. This difference manifests itself in the form of the topological term $S(p)$ \eqref{defSp}, which determines the renormalized vacuum expectation values (\ref{So3ren}) in the adiabatic approximation. The relationship between the $\lambda$-representation of the group $SO(3)$ and the Wigner function is found.

A separate issue, beyond the scope of this work, is the search for self-consistent solutions of the Einstein field equations with an EMT $\langle\hat{\tilde{T}}(\eta_{X},\eta_{Y})\rangle_{ren}$ in the adiabatic approximation.

\section*{Acknowledgements}

Breev and Shapovalov were partially supported by Tomsk State University under the International Competitiveness Improvement Program; Breev was partially supported by the Russian Foundation for Basic Research (RFBR) under the project No. 18-02- 00149; Shapovalov was partially supported by Tomsk Polytechnic University under the International Competitiveness Improvement Program and by RFBR and Tomsk region according to the research project No. 19-41-700004.

\end{document}